\documentclass[a4paper,12pt]{article}
\usepackage{latexsym,amssymb,amsfonts,amsmath,amsthm}
\usepackage[dvips]{graphicx}
\newtheorem{theorem}{Theorem}
\newtheorem{lemma}{Lemma}
\newtheorem{corollary}{Corollary}
\newtheorem{proposition}{Proposition}
\newtheorem{remark}{Remark}
\newtheorem{definition}{Definition}

\def\bra#1{\mathinner{\langle{#1}|}}
\def\ket#1{\mathinner{|{#1}\rangle}}
\def\braket#1{\mathinner{\langle{#1}\rangle}}

\renewcommand{\labelenumi}{(\theenumi)}

\renewcommand{\labelenumi}{(\theenumi)}

\title{{\Large {\bf Trapping and spreading properties of quantum walk in homological structure  
}
}}
\author{ 
{\small 
Takuya Machida,$^{1}$ 
\footnote{machida@stat.t.u-tokyo.ac.jp 
}\quad 
Etsuo Segawa,$^{2}$ 
\footnote{e-segawa@m.tohoku.ac.jp 
}\quad
}\\ 
{\scriptsize $^{1}$ 
Japan Society for the Promotion of Science, Japan}\\ 
{\scriptsize $^{1}$ 
Department of Mathematics, University of California, Berkeley
}\\
{\scriptsize 
California, Berkeley, 94720, USA
} \\
{\scriptsize $^3$ 
Graduate School of Information Sciences, Tohoku University, 
}\\
{\scriptsize 
Aoba, Sendai 980-8579, Japan
} \\
} 
\date{\empty }
\begin{document}
\maketitle

\par\noindent
\begin{small}
\par\noindent
{\bf Abstract}. 
We attempt to extract a homological structure of two kinds of graphs by the Grover walk. 
The first one consists of a cycle and two semi-infinite lines 
and the second one is assembled by a periodic embedding of the cycles in $\mathbb{Z}$. 
We show that both of them have essentially the same eigenvalues induced by the existence of cycles in the infinite graphs. 
The eigenspace of the homological structure appears as so called {\it localization} in the Grover walks, 
in that the walk is partially trapped by the homological structure. 
On the other hand, the difference of the absolutely continuous part of spectrum between them provides different behaviors. 
We characterize the behaviors by the density functions in the weak convergence theorem:  the first one is the delta measure at the bottom while 
the second one is expressed by two kinds of continuous functions which have different finite supports $(-1/\sqrt{10},1/\sqrt{10})$ and $(-2/7,2/7)$, respectively. 


\end{small}

\setcounter{equation}{0}

\section{Introduction}
The Grover walk arises from application of the quantum search algorithm ~\cite{Grover} to some spatial structures~\cite{Watrous} 
and accomplishes the quadratically speed up e.g.~\cite{Ambainis,Kempe,Shenvi}, compared to so called classical search algorithm. 
The Grover walks consist of an inherited eigenspace from the system of the simple random walk and a specific eigenspace of quantum walk~\cite{Sze}. 
A part of the effect of the inherited part on infinite graphs has been getting revealed gradually, for example linear spreading~\cite{Konno2002,Konno2005}, 
while the efficiency of the quantum search algorithm based on the Grover walk on finite graphs can be estimated by the hitting time 
of the simple random walk~\cite{Sze}. 
However the effect of the specific eigenspace of the quantum walk on the behavior has not been well investigated yet. 
Recently, it was shown that the specific eigenspace is deeply related to an underlying homological structure of the graph~\cite{HKSS_2014}. 
We expect that this eigenspace of the Grover walk plays an important role to recognize some spatial structures for an image processing engineering~\cite{VA1,VA2} in the future. 

Now let us consider the following model to extract its picture motivated by scattering theories~\cite{Exner,FH,Gnutzmann,SchanzSmilansky,Tanner}: 
for a given finite graph $G(V,E)$, we choose two vertices from $V(G)$, and attach two semi-infinite lines to 
the selected two vertices. We denote such a graph $\widetilde{G}$. 
Since our interest is a characterization of a homological structure of $G$ by the Grover walk, 
for simplicity we examine the $4$-length cycle $C_4$ as $G$ throughout this paper. 
Section 2 provides more precise definition. 
By the way, it is known that the Grover walk corresponds to the potential free Schr{\"o}dinger equation~\cite{Gnutzmann,HKSS_2013,Tanner}. 
As is the underlying Schr{\"o}dinger equation in the scattering situation, the Grover walk on the line gives a perfect transmissive behavior. 
Once a walker gets away from the original graph region, a particle gets farther away to infinity ballistically. 
However due to the existence of the structure $C_4$ on the line, we obtain a non-trivial observation as follows. 
Let $X_t$ be the position of a quantum walker on negative or positive semi-infinite lines at time $t\in \{0,1,2,\dots\}$. 
Let us consider two cases of the initial state:
	\begin{enumerate}
        \renewcommand{\labelenumi}{(\roman{enumi})}
	\item  from one (directed) edge in the negative semi-infinite line.
	\item from one (directed) edge in the original graph $C_4$.
	\end{enumerate}
We can compute reflection and transmission rates with respect to the region of $C_4$ as follows. 
	\begin{equation}\label{ballistic}
	\lim_{t\to\infty}\mathbb{P}\left(\frac{X_t}{t}\leq x\right)= \int_{-\infty}^x \left\{c_R\delta_{-1}(y)+c_O\delta_0(y)+c_T\delta_1(y)\right\} dy, 
	\end{equation}
where 
\begin{center}
	\begin{tabular}{c|ccc}
                 & $c_R$  & $c_O$ & $c_T$ \\ \hline 
        case (i) & $1/5$  & $0$   & $4/5$ \\
        case (ii) & $9/20$ & $1/2$ & $1/20$
        \end{tabular}
\end{center}
In case (ii), we observe that a part of particle is trapped in the region of the original graph.  
In this paper, we show that the important difference between the first case (i) and the second case (ii) consists in an 
overlap between their initial states and the ``homological eigenspace" defined as follows:
\begin{definition}\label{def}
For a closed cycle in $\widetilde{C}_4$, $c=((0',u),(u,0),(0,d),(d,0'))$, 
the homological eigenspace $\Gamma\subset \ell^2(A(\widetilde{C}_4))$ treated here is spanned by the following subspaces $\{\Gamma_m\}_{m=1}^4$: 
\begin{equation*}
	\Gamma_m=\mathrm{span}\{w^{(m)}(c)-w^{(m)}(\bar{c})\},\;\; m\in\{1,2,3,4\}, 
\end{equation*}
where $w^{(m)}: P(G)\to \ell^2(A)$ is given by, for a path $p=(e_1,e_2,\dots,e_n)\in P(G)$,  
	\begin{align*}
        w^{(m)}(p) = \sum_{j=1}^n e^{2\pi i m j/n} \delta_{e_j},\;\; m\in\mathbb{N}.
	\end{align*}
\end{definition}
\bigskip
A necessary and sufficient condition for the trap of a quantum walker in the region of $C_4$ is described as follows.
We call this phenomena localization. 
\begin{theorem}\label{ncc}
Localization happens in $\widetilde{C}_4$ if and only if 
the initial state $\Psi_0$ satisfies 
	\begin{equation*}
        \Psi_0\notin \Gamma^{\bot}.
	\end{equation*}
\end{theorem}
\bigskip
Instead of the attachment of the two semi-infinite lines, we take a periodic attachment of $C_4$. More precisely,  
we prepare infinite number of copies of $C_4$ labeled by 
\[\{\cdots,C_4^{(-2)},C_4^{(-1)},C_4^{(0)},C_4^{(1)},C_4^{(2)},\cdots\}. \]
We connect $V(C^{(j)})$ and $V(C^{(j+1)})$ by one edge (called bridge) for all $j\in \mathbb{Z}$. 
Section 2 provides more precise definition. 
\begin{theorem}
\label{TM:th:convergence}
In the above setting of the spatial structure with initial state $\Psi_0$, the following statements hold.
\begin{enumerate}
\item Localization happens if and only if 
the initial state $\Psi_0$ satisfies 
	\begin{equation*}
        \Psi_0\notin \left(\bigoplus_{j=-\infty}^{\infty}(j,\Gamma)\right)^{\bot}.
	\end{equation*}
\item Let $X_t$ be the label number of $C_4$'s copy of a quantum walk at time $t$. 
Then we have for any $x\in \mathbb{R}$, 
	\begin{align*}
 	\lim_{t\to\infty}\mathbb{P}\left(\frac{X_t}{t}\leq x\right)
        =\int_{-\infty}^x \left\{ \sum_{j=-\infty}^{\infty}||\Pi_{(j,\Gamma)}\Psi_0||^2\delta_0(y)+f(y)\right\} dy,
	\end{align*}
where 
$f$ is a linear combination of two continuous functions $f_1$ and $f_2$ which have the 
supports $(-1/\sqrt{10},\,1/\sqrt{10})$ and $(-2/7,\;2/7)$, respectively. 
Here $\Pi_\mathcal{H'}$ is the orthogonal projection onto the subspace $\mathcal{H'}$. 
\end{enumerate}
\end{theorem}
In the above spatial change, we observe that the periodic homological eigenspaces provides localization. 
On the other hand, the behavior of the remaining part is changed 
from the ballistic spreading in Eq.~(\ref{ballistic}) to a linear spreading whose density function has continuous supports; that is, 
the speed to infinite positions decreases. 
We define the ballistic spreading and linear spreading explicitly in section 2. 
We discuss these mechanisms in the rest of this paper. 

This paper is organized as follows. In section 2, we provide the definition of our models. The homological eigenspace of the Grover walk and the proof of Theorem~\ref{ncc} are presented in section 3. 
Section 4 is spent for the proof of Theorem~\ref{TM:th:convergence}. 
We propose a parametric expression to describe the density function of the limit distribution. 
Finally, we give a discussion in section 5.  
%
\section{Definition}
In this paper, we treat two kind of infinite graphs. Both of them are constructed from $C_4$. 
Let the vertices and edges of $C_4$ be labeled by: 
	\[ V=\{0,u,0',d\} \;\mathrm{and}\; E=\{\{0,u\},\{u,0'\},\{0',d\},\{d,0\}\}. \]
Take two infinite half lines  $H_\pm(V_\pm,E_\pm)$, $V_+=\{1,2,3,\dots\}$ and $V_-=\{-1,-2,-3,\dots\}$ with $E_+=\{\{1,2\},\{2,3\},\dots\}$ and $E_-=\{\{-1,-2\},\{-2,-3\},\dots\}$, respectively. 
The first graph $\widetilde{C}_4$ is defined as follows (see Fig.~1). 
We connect the two semi-infinite lines $H_+$ and $H_-$ to $C_4$ by edges: 
	\begin{align*}
	V(\widetilde{C}_4) &=V(C_4)\cup V_+ \cup V_-, \\
	E(\widetilde{C}_4) &= E(C_4)\cup E_+\cup E_-\cup \{\{0',-1\},\{0,1\}\}. 
	\end{align*}
The second graph $C_4'$ is defined as follows (see Fig.~2). 
We take infinite number of copies of $C_4$ labeled by $\{C_4^{(j)}\}_{j\in \mathbb{Z}}$. 
The vertices and edges of $C_4^{(j)}$ are labeled by $\{0_j,u_j,0_j',d_j\}$ and $\{\{0_j,u_j\},\{u_j,0_j'\},\{0_j',d_j\},\{d_j,0_j\}\}$, respectively. 
We connect $C_4^{(j)}$ and $C_4^{(j+1)}$ by one edge for all $j\in \mathbb{Z}$: 
	\begin{align*}
	V(C_4') &=\bigcup_{j\in \mathbb{Z}}V(C_4^{(j)}), \\
	E(C_4') &=\bigcup_{j\in \mathbb{Z}} \left(E(C_4^{(j)}) \cup \{0_j,0_{j+1}'\} \right) . 
	\end{align*}
For a connected undirected graph $G$, let $A(G)$ be the set of arcs induced by edge of $G$ such that $A(G)=\{(u,v)\in V(G)\times V(G): \{u,v\}\in E(G)\}$. 
The arc $e=(u,v)$ is regarded as directed edge from $u$ to $v$; that is, $o(e)=u$ and $t(e)=v$. 
The total state of the quantum walk on $G\in \{\widetilde{C}_4,C_4'\}$ is described by a Hilbert space
	\begin{equation*}
	\mathcal{H}_G=\ell^2(A(G)). 
	\end{equation*}
For $\psi,\phi\in \ell^2(A(G)$, 
we use the notation $\langle \psi, \phi \rangle=\sum_{e\in A(G)} \bar{\psi}(e)\phi(e)$ as the inner product, and $|| \psi ||^2=\langle \psi, \psi \rangle$ as the norm. 
We employ the standard basis $\{\delta_e : e\in A(G)\}$ for $\mathcal{H}_G$, where for $f\in A(G)$, $\delta_e(f)=1$ $(e=f)$, $=0$ (otherwise). 
For a subspace $\mathcal{H}'\subset \mathcal{H}_G$, we define $\Pi_{\mathcal{H}'}$ as the orthogonal projection onto $\mathcal{H}'$. 
The time evolution is determined by a unitary operator $U_G: \ell^2(G)\to \ell^2(G)$ defined as follows: 
	\begin{equation*}
	\langle \delta_f,  U_G \delta_e\rangle = \left(\frac{2}{\deg{o(e)}}-\delta_{e,\bar{f}}\right) \boldsymbol{1}_{\{o(e)=t(f)\}}(e,f), 
	\end{equation*}
where $\deg{u}$ is the degree of the vertex $u$. 
For given initial state $\Psi_0\in \ell^2(A(G))$ with $||\Psi_0||^2=1$, we consider the iteration of the unitary operation; $\Psi_0 \mapsto \Psi_1 \mapsto \Psi_2 \mapsto \cdots$, where 
$\Psi_t=U_G\Psi_{t-1}$. 
The unitarity of the operator $U_G$ preserves the norm. So we can define the distribution at each time  $\mu_t: V(A(G))\to [0,1]$ such that 
	\begin{equation*}
	\mu_t(u)=\sum_{e:o(e)=u}|\Psi_t(e)|^2.
	\end{equation*}
This is regarded as the finding probability of a quantum walker at the vertex $u\in V(G)$ at time $t$ in this paper. 
We focus on the following limit behaviors proposed by this paper. 
\begin{definition}
Let $\{\rho_t\}_{t=0}^{\infty}$ be a sequence of distributions on $\mathbb{Z}$ satisfying $\rho_t(j)=0$ for any $j$ such that $|j|>ct$ $(t=1,2,\dots)$. 
Here the value $c$ is a positive constant. 
\begin{enumerate}
\renewcommand{\labelenumi}{(\Roman{enumi})}
\item 
If there exists a finite integer $j$ such that 
	\[ \limsup_{t\to\infty} \rho_t(j)>0,  \]
then we say {\it localization} occurs in the sequence of $\{\rho_t\}_{t=0}^{\infty}$. \\
In particular, if it holds that 
	\[  \left|\left\{j\in \mathbb{Z}: \limsup_{t\to\infty} \rho_t(j)>0\right\}\right|\in (0,\infty),  \]
then we say {\it strong localization} occurs. 
\item  
Assume that there exists a right-continuous function $F$ on $\mathbb{R}$ such that for $x\in \mathbb{R}$,
	\[ \lim_{t\to\infty} \sum_{j<tx}\rho_t(j)=F(x). \]
\begin{enumerate}
\item 
If there exists $\alpha,\beta\in [-c,c]$ ($\alpha<\beta$) such that $F\in C^1$ and $dF(x)/dx>0$ on the interval $(\alpha,\beta)$, 
then we say {\it linear spreading (with continuous support)} occurs. 
\item 
If $F$ has discontinuities on $[-c,c]$ except the origin, 
then we say {\it ballistic spreading} occurs.  
\end{enumerate}
\end{enumerate}
\end{definition}
Throughout this paper, a random variable $X_t$ at time $t$ follows 
	\[ \mathbb{P}(X_t=j)=\sum_{v\in V_j}\mu_t(v)\;\;(j\in \mathbb{Z}), \]
where for $\widetilde{C}_4$ case, 
	 \[V_j=\begin{cases} \{j\} & \text{: $j\neq 0$,} \\ V(C_4) & \text{: $j=0$}, \end{cases} \]
and for $C_4'$ case, $V_j=V(C_4^{(j)})$. 
\section{Homological eigenspace}
In this section, we prove the necessary sufficient condition of the localization of the Grover walk on $\widetilde{C}_4$ in Theorem \ref{ncc}. 
Figure 1 is useful to find out the following statement. 
We should remark that on the vertices whose degree are two, a quantum walker takes a trivial motion; if a particle came from right (resp. left) direction, 
then in the next step it goes to right (resp. left) without turn to the opposite direction. 
So it is sufficient to consider the following initial state $\Psi_0\in \ell^2(A(\widetilde{C}_4))$; for every $e\in$  $\{$arcs of positive and negative half lines except $(0',-1)$ and $(0,1)$ $\}$, 
\[ \langle \delta_e, \Psi_0 \rangle=0. \]
See Figure~\ref{fig:space2} for the initial state. 
Then we can compute the behavior of the Grover walk explicitly in the following lemma. 
In the computation, we just pay attention to the only two-exception; $0$ and $0'$ where both reflection and transmission happen. 
\begin{figure}[h]
 \begin{center}
  \includegraphics[scale=1.2]{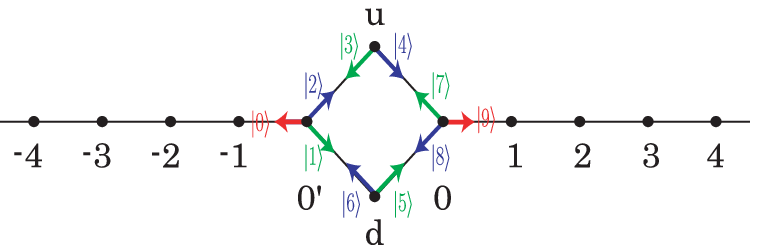}
 \end{center}
 \caption{Space and coin-state for the Grover walk on $\widetilde{C}_4$}
 \label{fig:space2}
\end{figure}
\begin{lemma}\label{machida's lem}
For complex numbers $a_j$ (j=0,1,\dots,9) with $\sum_{j=0}^9|a_j|^2=1$, we take the initial condition to be
\begin{align}
 \ket{\Psi_0}=&\ket{0'}\otimes(a_0\ket{0}+a_1\ket{1}+a_2\ket{2})+\ket{u}\otimes (a_3\ket{3}+a_4\ket{4})\nonumber\\
 &+\ket{d}\otimes (a_5\ket{5}+a_6\ket{6})+\ket{0}\otimes(a_7\ket{7}+a_8\ket{8}+a_9\ket{9}).
\end{align}
Then we have 
\begin{align}
 \mu_n(-1)=&\left\{\begin{array}{ll}
	     \frac{1}{9}|a_0-2a_1-2a_2|^2 & (n=1)\\[1mm]
		    \frac{4}{9}|a_4+a_5|^2 & (n=2)\\[1mm]
		    \frac{4}{81}|a_7+a_8+4a_9|^2 & (n=3)\\[1mm]
		    \frac{4}{81}|a_3+a_6|^2 & (n=4)\\[1mm]
		    \frac{4}{9^{2m+1}}|4a_0+a_1+a_2|^2 & (n=4m+1 ; m=1,2,\ldots)\\[1mm]
		    \frac{4}{9^{2m+1}}|a_4+a_5|^2 & (n=4m+2 ; m=1,2,\ldots)\\[1mm]
		    \frac{4}{9^{2m+2}}|4a_7+a_8+a_9|^2 & (n=4m+3 ; m=1,2,\ldots)\\[1mm]
		    \frac{4}{9^{2m+2}}|a_3+a_6|^2 & (n=4m+4 ; m=1,2,\ldots)
		  \end{array}\right., \label{neg} \\[2mm]
 \mu_n(1)=&\left\{\begin{array}{ll}
	     \frac{1}{9}|2a_7+2a_8-a_9|^2 & (n=1)\\[1mm]
		    \frac{4}{9}|a_3+a_6|^2 & (n=2)\\[1mm]
		    \frac{4}{81}|4a_0+a_1+a_2|^2 & (n=3)\\[1mm]
		    \frac{4}{81}|a_4+a_5|^2 & (n=4)\\[1mm]
		    \frac{4}{9^{2m+1}}|4a_7+a_8+a_9|^2 & (n=4m+1 ; m=1,2,\ldots)\\[1mm]
		    \frac{4}{9^{2m+1}}|a_3+a_6|^2 & (n=4m+2 ; m=1,2,\ldots)\\[1mm]
		    \frac{4}{9^{2m+2}}|4a_0+a_1+a_2|^2 & (n=4m+3 ; m=1,2,\ldots)\\[1mm]
		    \frac{4}{9^{2m+2}}|a_4+a_5|^2 & (n=4m+4 ; m=1,2,\ldots)
		  \end{array}\right., \label{pos}
\end{align}
and
\begin{align}
 \lim_{n\to\infty}\mu_{4n+j}(0')=&\left\{\begin{array}{ll}
			      \frac{1}{2}|a_1-a_2|^2 & (j=1)\\[1mm]
			      \frac{1}{2}|a_4-a_5|^2 & (j=2)\\[1mm]
			      \frac{1}{2}|a_7-a_8|^2 & (j=3)\\[1mm]
			      \frac{1}{2}|a_3-a_6|^2 & (j=4)
		  \end{array}\right.,\label{origin} \\[2mm] 
 \lim_{n\to\infty}\mu_{4n+j}(0)=&\left\{\begin{array}{ll}
			      \frac{1}{2}|a_7-a_8|^2 & (j=1)\\[1mm]
			      \frac{1}{2}|a_3-a_6|^2 & (j=2)\\[1mm]
			      \frac{1}{2}|a_1-a_2|^2 & (j=3)\\[1mm]
			      \frac{1}{2}|a_4-a_5|^2 & (j=4)
		  \end{array}\right., \label{origingin}
\end{align}
\begin{equation} \label{ud}
  \lim_{n\to\infty}\mu_{4n+j}(u)=\lim_{n\to\infty}\mu_{4n+j}(d)=\left\{\begin{array}{ll}
								 \frac{1}{4}(|a_3-a_6|^2+|a_4-a_5|^2) & (j=1,3)\\[2mm]
								 \frac{1}{4}(|a_1-a_2|^2+|a_7-a_8|^2) & (j=2,4)
								       \end{array}\right..
\end{equation}
\end{lemma}
From Eqs.~(\ref{origin})-(\ref{ud}) in Lemma~\ref{machida's lem}, we observe that 
there are no contributions of $a_0$ and $a_9$, which are the initial amplitudes assigned to the arcs of negative and positive half lines, to the localization at the region of $C_4$. 
The possible region of the graph which exhibits localization is the subgraph $C_4$ since 
from Eqs.~(\ref{neg}) and (\ref{pos}) in Lemma~\ref{machida's lem}, we have 
	\begin{equation*}
	\lim_{n\to\infty}\mu_n(v)=0\;\;\; (v\notin\{0,0',u,d\}). 
	\end{equation*}
So the subspace spanned by the arcs associated with the two half lines has no contribution to the localization at $\widetilde{C}_4$. 
What is the essential subspace to provide the strong localization? 
The answer is the ``homological eigenspace" in Definition~\ref{def}. 
Recall that $c\in P(\widetilde{C}_4)$ is the following closed cycle of $\widetilde{C}_4$. 
	\[ c=((0',u),(u,0),(0,d),(d,0')).\] 
We denote $\bar{c}$ as the inverse closed cycle of $c$. 
The  ``eigenspace'' comes from the following fact. 
\begin{remark}
For $m\in\{0,1,2,3\}$, 
$w^{(m)}(c)-w^{(m)}(\bar{c})\in \ell^2(A(\widetilde{C}_4))$ is the eigenvector of eigenvalue $i^{m}$ for the time evolution of the Grover walk $U_{\widetilde{C}_4}$. 
\end{remark}
\bigskip
Thus the subspace $\Gamma_m$'s are the eigenspaces of the system. 
From Eqs.~(\ref{origin})-(\ref{ud}), we can express the summation of the limit distribution over all the positions 
by using the eigenspaces $\Gamma_m$ $(m\in\{0,1,2,3\})$ as follows: 
	\begin{equation*}
        \sum_{v\in V(\widetilde{C}_4)}\lim_{n\to\infty}\mu_n(v)=\sum_{m=0}^3 ||\Pi_{\Gamma_m}\Psi_0  ||^2.
	\end{equation*}
Indeed, 
	\begin{align*}
        ||\Pi_{\Gamma_m}\Psi_0||^2 &= \frac{1}{8}\left|(a_2-a_1)+(i)^m(a_4-a_5)+(i)^{2m}(a_8-a_7)+(i)^{3m}(a_6-a_3) \right|^2 \\
        		      &= \frac{1}{2} || \Pi_{\{s_m\}}\psi_0 ||_{\mathbb{C}^4}^2, 
        \end{align*}
where $s_m={}^T[1/2,i^m/2,i^{2m}/2,i^{3m}/2]$, and $\psi_0={}^T[a_2-a_1, a_4-a_5, a_8-a_7,a_6-a_3]$. 
Remarking that $\{s_m\}_{m=0}^3$ is a complete orthonormal base of $\mathbb{C}^4$, 
	\begin{equation*}
        \frac{1}{2} \sum_{m=0}^3|| \Pi_{\{s_m\}}\psi_0 ||^2_{\mathbb{C}^4}= \frac{1}{2}||\psi_0||_{\mathbb{C}^4}^2. 
        \end{equation*}
This is equivalent to the summation of $\lim_{n\to\infty}\mu_n(v)$ over $v\in\{0,u,0',d\}$ described by Eqs.~(\ref{origin})-(\ref{ud}). 

We conclude that the overlap between the initial state and the homological eigenspaces is a necessary and sufficient condition for localization, 
which implies the desired conclusion of Theorem \ref{ncc}. 

\section{Spectral analysis and limit behavior for the Grover walk on $C_4'$}
\subsection{Spectral mapping theorem}
Let $\mathcal{D}$ be the fundamental domain of the graph $C_4'$~\cite{Sunada}. 
The fundamental domain $\mathcal{D}$ is represented by $V(\mathcal{D})=V(C_4)$, $E(\mathcal{D})=E(C_4)\cup \{0,0'\}$. 
We assign a one form $\theta$ to each arc in $A(\mathcal{D})$ such that 
	\begin{equation*}
	\theta(f)=
        \begin{cases}
        k & \text{: $f=(0,0')$,}\\
        -k &  \text{: $f=(0',0)$,}\\
        0 & \text{: otherwise.}
        \end{cases}
	\end{equation*}
Since it holds that 
$V(C_4')\cong \mathbb{Z}\times V(\mathcal{D})$, $A(C_4')\cong \mathbb{Z}\times A(\mathcal{D})$\footnote{ 
The one-to-one correspondence between $V(C_4')$ and $\mathbb{Z}\times V(\mathcal{D})$ is denoted by 
$v_j\leftrightarrow (j,v)$,  $(v\in \{0,0',u,d\})$ and 
one between $A(C_4')$ and $\mathbb{Z}\times A(\mathcal{D})$ is $(v_j,w_j)\leftrightarrow (j,(v,w))$ for $(v_j,w_j)\in A(C_4^{(j)})$, 
$(0_j,0'_{j+1})\leftrightarrow (j,(0,0'))$ and $(0_j',0_{j-1})\leftrightarrow (j,(0',0))$. },  
we regard the vertices and the arcs of $C_4'$ as the elements of $\mathbb{Z}\times V(\mathcal{D})$ and $\mathbb{Z}\times A(\mathcal{D})$. 

Now we take the Fourier transform $\mathcal{F}: \ell^2(\mathbb{Z}\times A(\mathcal{D}))\to L^2([-\pi,\pi)\times A(\mathcal{D}))$ defined by 
	\begin{equation*}
  	 \hat\psi(k,f)\equiv\mathcal{F}(\psi)(k,f)=\sum_{x \in \mathbb{Z}} e^{-ikx}\psi(x,f) \quad (k\in [-\pi,\pi)).
	\end{equation*}
The Fourier inversion transform $\mathcal{F}^{-1}: L^2([-k,k)\times A(\mathcal{D})) \to \ell^2(\mathbb{Z}\times A(\mathcal{D}))$ is 
	\begin{equation*}
  	\mathcal{F}^{-1}(\hat{\psi})(x,f)= \int_{-\pi}^\pi e^{ikx}\hat{\psi}(k) \frac{dk}{2\pi} \quad (x\in \mathbb{Z}). 
	\end{equation*}
The time evolution $\Psi_{t}=U_{C_4'}^t \Psi_0$ makes a relationship  
	\begin{equation*}
 	\hat\Psi_{t+1}(k)=\hat{U}(k)\hat\Psi_t(k),
	\end{equation*}
where for $\psi\in L^2([-\pi,\pi)\times A(\mathcal{D}))$, 
	\[ (\hat{U}(k)\psi)(k,f)=  \sum_{e:o(e)=t(f)} e^{-i\theta(f)}\left( \frac{2}{\deg (o(e))}-\delta_{e,\bar{f}} \right) \psi(k,e). \]
For fixed $k$, we regard $\hat{U}(k)$ as a unitary operator on $\ell^2(\mathcal{D})$ whose inner product is 
$\langle \psi,\phi \rangle_{\mathcal{D}}=\sum_{e\in A(\mathcal{D})}\bar{\psi}(e)\phi(e)$. 
Such a quantum walk iterated by $\hat{U}(k)$ is called {\it twisted} Szegedy walk~\cite{HKSS_2014}. 
On the other hand, for fixed $k$, the {\it twisted} random walk $P(k): \ell^2(V(\mathcal{D}))\to \ell^2(V(\mathcal{D}))$ is defined as follows~\cite{Sunada}: 
	\begin{equation}\label{tRW}
	P(k) \cong 
        \begin{bmatrix}
        0 & 1/2 & 1/2 & e^{ik}/3 \\
        1/3 & 0 & 0 & 1/3 \\
        1/3 & 0 & 0 & 1/3 \\
        e^{-ik}/3 & 1/2 & 1/2 & 0
        \end{bmatrix}. 
        \end{equation}
Here we have taken 
$\delta_{0'}\cong {}^T[1,0,0,0]$, $\delta_u\cong {}^T[0,1,0,0]$, $\delta_d\cong {}^T[0,0,1,0]$ and $\delta_{0}\cong {}^T[0,0,0,1]$. 
We assign transition probability $p: A(\mathcal{D})\to [0,1]$ such that 
	\begin{equation*}
        p(e)=
        \begin{cases}
        1/3 & \text{: $o(e)\in \{0,0'\}$,} \\
        1/2 & \text{: $o(e)\in \{u,d\}$.}
        \end{cases}
        \end{equation*}
Moreover we define the reversible probability $\pi: \ell^2(V)\to \ell^2(V)$ such that 
	\begin{equation*}
        \pi(0)=\pi(0')=3/10,\;\; \pi(u)=\pi(d)=1/5.
        \end{equation*}
Let $A,B: \ell^2(V(\mathcal{D}))\to \ell^2(A(\mathcal{D}))$ be 
	\begin{align*} 
        A\delta_v &= \sum_{e:o(e)=v}\sqrt{p(e)}\delta_{e}, \\
        B\delta_v &= \sum_{e:t(e)=v}e^{-i\theta(e)}\sqrt{p(\bar{e})}\delta_{e}.   
        \end{align*}
For $z,w\in \mathbb{C}$, we define $\Phi_{z,w}: \ell^2(V)\to \ell^2(A)$ such that 
	\[ \Psi_{z,w}(f)=zAf+wBf. \]
According to Higuchi et al.~\cite{HKSS_2014}, we have the spectrum of the twisted Szegedy walk as follows. 
\begin{lemma}\label{specmap}
\noindent 
\begin{enumerate}
\item Let $\varphi_{QW}(x)=(x+x^{-1})/2$. Then 
	\begin{equation}\label{spec}
	\mathrm{spec}(\hat{U}(k))=\varphi_{QW}^{-1}(\mathrm{spec}(\hat{P}(k)))\cup \{1\}^{1+\delta_0(k)} \cup \{-1\}^{1+\delta_\pi(k)}.
	\end{equation}
\item 
The eigenvector $w_\alpha$ of the eigenvalue $e^{i\alpha}$ in the first term of RHS in Eq.~(\ref{spec}) are described as follows. 
Let $f_\lambda\in \ell^2(V)$ be the eigenvector satisfying $\hat{P}(k)f_\lambda=\lambda f_\lambda$. 
Then 
	\begin{equation*} 
         w_\alpha= \Phi_{1,-e^{i\alpha}}(D^{-1/2}_\pi f_{\cos \alpha}),
	\end{equation*}
where $(D_\pi f)(v)=\pi(v)f(v)$. 
\item 
The Fourier inversion of the eigenspaces of second and third terms of RHS in Eq.~(\ref{spec}) are 
	\begin{align*}
        \mathcal{F}^{-1}(\mathcal{M}_+) &\cong (0, \Gamma_0), \\
        \mathcal{F}^{-1}(\mathcal{M}_-) &\cong  (0, \Gamma_2). 
        \end{align*}
\end{enumerate}      
\end{lemma}
\bigskip
The characteristic polynomial of $\hat{P}(k)$ is obtained as follows. 
	\begin{equation*}
        \mathrm{det}(\lambda - \hat{P}(k))=\frac{\lambda}{9}(9\lambda^3-7\lambda-2\cos k). 
        \end{equation*}
We find a constant eigenvalue $\lambda=0$ and the other three satisfying the following cubic equation. 
	\begin{equation}\label{cubic}
        9\lambda^3-7\lambda-2\cos k=0. 
        \end{equation}
In the following first and second subsections, we discuss the contribution of the constant eigenvalue $\lambda=0$ 
and the eigenvalues satisfying Eq.~(\ref{cubic}) to the behavior of the QW, respectively.  
\subsection{Homological eigenspace of $C_4'$}
In this subsection, we clarify the eigenspace associated with the constant eigenvalues with respect to $k$. 
Recall that the unitary matrix $\hat{U}(k)$ has such three eigenvalues $-1$, $0$ and $1$. 
The eigenvalues $\pm 1$ come from the genesis part of quantum walk (corresponding to second and third terms of RHS in Eq.~(\ref{spec})). 
It has already found out by Lemma~\ref{specmap} that the eigenspaces are $\mathbb{Z}\times \Gamma_0$ and 
$\mathbb{Z}\times \Gamma_2$, respectively. 
Now we consider the eigenvalue $0$ which comes from the inherited from the twisted random walk. 
The eigenvector $f_0\in \ell^2(V)$ is expressed by 
	\[ f_0={}^T[0,1,-1,0]. \]
From Lemma~\ref{specmap}, the corresponding eigenvalues of $\hat{U}(k)$ are $\pm i$, and the eigenvectors are interestingly 
$w_1(c)-w_1(\bar{c})$ and $w_3(c)-w_3(\bar{c})$, respectively. Then we have 
	\[ \mathrm{span}\{\mathcal{F}^{-1}[e^{ixk}(w_m(c)-w_m(\bar{c}))]: x\in \mathbb{Z}\}= \mathbb{Z}\times \Gamma_m \;\;(m\in\{ 1,3\}). \]
We conclude that the homological eigenspace of $\tilde{C}_4$ periodically exists in $C_4'$. 
We summarize these statements obtained in this subsection. 
\begin{proposition}\label{eigenspaceC_4'}
All the eigenvalues of the Grover walk on $C_4'$ are $\{i^m;m\in\{0,1,2,3\}\}$. 
The eigenspace $\Gamma'_m$ is described by 
	\[ \Gamma'_m \cong \mathbb{Z}\times \Gamma_m \;\;\;\;(m\in\{0,1,2,3\}). \]
\end{proposition}
\begin{remark}
An expression for a complete orthonormal base; $\{\eta_{m;j}\}_{j\in\mathbb{Z}}$, of the eigenspace $\Gamma_m'$ is expressed as follows: 
	\begin{equation*}
        \left\{ \frac{1}{\sqrt{8}}\left(w_m(c_j)-w_m(\overline{c_j})\right) \right\}_{j\in \mathbb{Z}}.
        \end{equation*}
Here $c_j$ is the closed cycle at the unit of $j$ (see Fig.~\ref{fig:space}); 
	\[c_j=\left((j,|2\rangle),(j,|4\rangle),(j,|8\rangle),(j,|6\rangle)\right). \]
\end{remark}
\begin{figure}[h]
 \begin{center}
  \includegraphics[scale=0.6]{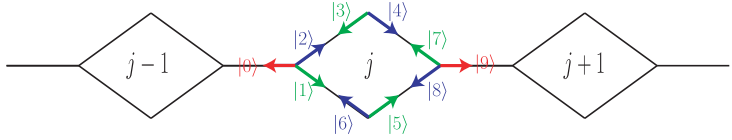}
 \end{center}
 \caption{Space and coin-state for the Grover walk on $C_4'$}
 \label{fig:space}
\end{figure}
As a consequence of Proposition~\ref{eigenspaceC_4'}, we observe that the localization 
happens if and only if the initial state $\Psi_0$ has an overlap to the periodic homological eigenspaces 
$\mathbb{Z}\times \Gamma_m$ $(m\in\{0,1,2,3\})$. 
Indeed, 
	\[ \mu_n(v) \sim \sum_{e:o(e)=v}\left|\sum_{m=0}^3(i^{mn}\Pi_{\Gamma'_m}\Psi_0)(e)\right|^2, \]
and the mass of the delta measure in the weak limit theorem is 
	\begin{equation}\label{Delta} 
        \Delta=\sum_{v\in V(C_4')}\lim_{n\to\infty}\mu_n(v)=\sum_{j\in \mathbb{Z}}\sum_{m=0}^3|\langle \eta_{m;j},\Psi_0\rangle|^2. 
        \end{equation}
Therefore when the support of $\Psi_0$ is finite; that is, $\sum_{e\in A(\mathcal{D})}|\{j\in \mathbb{Z}: \Psi_0(j,e)\neq 0\}|<\infty$, 
then we can observe that a quantum walker is trapped in finite number of subgraphs $C_4^{(j)}$'s at rate $\Delta\in[0,1]$, 
where for a set $A$, $|A|$ is the cardinality of $A$. 
This is nothing but the strong localization. 
\subsection{Weak convergence theorem}
In this section we discuss the weak convergence for the quantum walk.
The solutions for the cubic equation Eq.~(\ref{cubic}) are obtained as follows. 
	\begin{equation}\label{randomwalk}
 	\lambda_j(k)=\frac{2\sqrt{7}}{3\sqrt{3}}\cos\left\{\frac{1}{3}\arccos\left(\frac{9\sqrt{3}}{7\sqrt{7}}\cos k\right) + \frac{2j\pi}{3}\right\}\,\;\;(j=0,1,2). 
	\end{equation}
Thus Lemma~\ref{specmap} provides the inherited six-eigenvalue from the twisted random walk Eq.~(\ref{tRW}) which depends on the wave number $k$: 
	\begin{equation*}
 	\nu_{j,l}(k)=\lambda_j(k)+(-1)^l\, i\sqrt{1-\lambda_j(k)^2}\quad (j=0,1,2,\,l=0,1).
	\end{equation*}
To get the formal representation in Theorem~\ref{TM:th:convergence} and the properties of the function $f(x)$, we compute the $r$-th moment $\mathbb{E}(X_t^r)$ ($r=0,1,2,\ldots$) according to the Fourier analysis which was introduced by Grimmett et al.~\cite{GrimmettJansonScudo2004}.
Let $\ket{v_{j,l}(k)}\,(j=0,1,2,\,l=0,1)$ (resp. $\ket{v_j^{(c)}(k)}\,(j=0,1,2,3)$) be normalized eigenvectors of the matrix $\hat{U}(k)$ corresponding to the eigenvalues $\nu_{j,l}(k)$ (resp. $\nu_j^{(c)}$).
The $r$-th moment on the Fourier space becomes
	\begin{align}
 	\mathbb{E}(X_t^r)=&\sum_{x\in\mathbb{Z}}x^r\mathbb{P}(X_t=x)\nonumber\\
 	=&\int_{-\pi}^\pi \bra{\hat\Psi_t(k)}\left(D^r\ket{\hat\Psi_t(k)}\right)\frac{dk}{2\pi}\nonumber\\
 	=&(t)_r \left\{ 0^r\Delta + \int_{-\pi}^\pi \sum_{j=0}^2\sum_{l=0}^1 \left(\frac{i\nu'_{j,l}(k)}{\nu_{j,l}(k)}\right)^r\left|\braket{v_{j,l}(k)|\hat\Psi_0(k)}\right|^2\frac{dk}{2\pi}\right\} \\ 
        &\qquad\qquad\qquad\qquad\qquad\qquad\qquad\qquad\qquad\qquad\qquad+O(t^{r-1}), \notag
	\end{align}
where $D=i(d/dk)$, $(t)_r=t(t-1)\times\cdots\times(t-r+1)$
and
	\begin{equation*}
 	\Delta=\int_{-\pi}^\pi \sum_{j=0}^3  \left|\braket{v_j^{(c)}(k)|\hat\Psi_0(k)}\right|^2\frac{dk}{2\pi}.
	\end{equation*}
As $t\to\infty$, we get a limit with respect to a rescaled random variable $X_t/t$,
	\begin{multline}\label{eq:r-th_moment}
 	\lim_{t\to\infty}\mathbb{E}\left[\left(\frac{X_t}{t}\right)^r\right]  \\
 	=\int_{-\infty}^\infty x^r \Delta\delta_0(x)\,dx
 	+\int_{-\pi}^\pi \sum_{j=0}^2\sum_{l=0}^1  \left(\frac{i\nu'_{j,l}(k)}{\nu_{j,l}(k)}\right)^r\left|\braket{v_{j,l}(k)|\hat\Psi_0(k)}\right|^2\frac{dk}{2\pi}. 
	\end{multline}
Here, putting $x_{j,l}(k)=i\nu'_{j,l}(k)/\nu_{j,l}(k)\,(j=0,1,2,\,l=0,1)$ for the second term in Eq.~(\ref{eq:r-th_moment}), we get an integral form of the term
	\begin{equation}\label{eq:2nd_term}
 	\int_{-\pi}^\pi \sum_{j=0}^2\sum_{l=0}^1  x_{j,l}(k)^r\left|\braket{v_{j,l}(k)|\hat\Psi_0(k)}\right|^2\frac{1}{\frac{dx_{j,l}(k)}{dk}}\,\frac{dx_{j,l}(k)}{2\pi}.
	\end{equation}
Straightforwardly computing the function $x_{j,l}(k)$, we get
	\begin{align}
 	x_{j,l}(k)=&(-1)^l\frac{\lambda'_j(k)}{\sqrt{1-\lambda_j(k)^2}}\nonumber\\
 	=&-(-1)^l \frac{2\sin k}{7\sqrt{1-A^2\cos^2 k}}\cdot\frac{\sin{\xi_j(k)}}{\sqrt{1-\frac{28}{27}{\cos^2\xi_j(k)}}},\label{eq:x(k)}\\
 	\frac{dx_{j,l}(k)}{dk}=&-(-1)^l\frac{A}{21}F({\cos\xi_j(k)})\sqrt{1-\frac{28}{27}{\cos^2\xi_j(k)}},\label{eq:dx(k)/dk}
	\end{align}
where
	\begin{align}
 	A=&\frac{9\sqrt{3}}{7\sqrt{7}},\\
 	\xi_j(k)=&{\frac{1}{3}\arccos(A\cos k)+\frac{2j\pi}{3}}\quad (j=0,1,2),\\
 	F(x)=&-\frac{2x(28x^2+33)}{9(4x^2-1)^3}. \label{F}
	\end{align}

We have seen that the constant eigenvalues cause localization on the quantum walk and appears as the coefficient $\Delta$ of the delta measure in Eq.~(\ref{eq:r-th_moment}).
On the other hand, the continuous functions $\nu_{j,l}(k)$ build the function $f(x)$.
Although it is hard to get the function $f(x)$ explicitly, 
from now on we discuss an abstractive shape of the function $f(x)$ by using the above computations. 
As a preparation, we provide a following argument. 
Let $h(k)$ and $g(k)$ be periodic and bounded functions with $h(k+2\pi)=h(k)$ and $g(k+2\pi)=g(k)$. 
We consider the interval {$[k_0,2\pi+k_0)$ as $\cup_{j=0}^{s-1} [k_j,k_{j+1})\,\,(k_0<\cdots<k_{s-1}<2\pi+k_0=k_s)$ 
so that the function $h(k)$ is a strictly monic and continuous function on $[k_j,k_{j+1})$ and at the boundaries
\footnote{{More precisely, we impose the following assumptions to $h(k)$. (i) $h(k)=h(k+2\pi)$ for all $k\in \mathbb{R}$, (ii) we permit discontinuity of $h(k)$ only at $\{2\pi n+k_j\}_{j=0}^{s-1}$, $n\in \mathbb{N}$. 
(iii) for any interval, $h(k)$ does not take a constant value.} },  
	\[ \lim_{\delta\downarrow 0} \frac{h(k_{j+1}+\delta)-h_{j+1}^{(R)}} {\delta}
        	=\lim_{\delta\uparrow 0} \frac{h(k_{j+1}+\delta)-h_{j+1}^{(L)}}{\delta}=0, \]
for all $j\in \mathbb{Z}/s\mathbb{Z}$.
Here $h_{j}^{(R)}=h(k_j)$ and $h_j^{(L)}=\lim_{k\uparrow k_{j+1}}h_j(k)$.
Then we have 
	\begin{align}\label{chikan}
        \int_{0}^{2\pi} h^r(k)g(k)dk=&\sum_{{j=0}}^{s-1} \int_{k_j}^{k_{j+1}} h^r(k)g(k)dk \nonumber\\
	 =&\int_{-\infty}^{\infty} x^r \sum_{{j=0}}^{s-1} \rho_j(x) dx, 
        \end{align}
\noindent where $\rho_j(x)$ has a finite support $(m_j,M_j)$ and the orbit of $(x,\rho_j(x))$ is expressed as follows: 
	\begin{equation}\label{parameter}
        \left\{(x,\rho_j(x)): x\in (m_j,M_j)\right\}= \left\{ \left( h(k), \frac{g(k)}{|h'(k)|} \right): k\in(k_j,k_{j+1}) \right\}. 
        \end{equation}
Since the function $h(k)$ is a one to one map on the interval $(k_j,k_{j+1})$, we determine a unique value $k\,\in(k_j,k_{j+1})$ such that $h(k)=x\,\in (m_j,M_j)$.
Equivalently another expression for $\rho_j(x)$ is 
	\begin{equation}\label{direct}
        \rho_j(x)=\frac{g(h^{-1}(x))}{|h'(h^{-1}(x))|} \boldsymbol{1}_{(m_j, M_j)}(x). 
        \end{equation}
Here 
	\[ h'(k)=
        \begin{cases} 
        d h(k)/d k & \text{: $k\notin\{{k_{0}},\dots,k_{{s-1}}\}$,} \\ 
        	0 & \text{: $k\in\{{k_{0}},\dots,k_{{s-1}}\}$,}
        \end{cases}\] 
and $m_j=\mathrm{min}\{h_j^{(R)},h_{j+1}^{(L)}\}$, 
$M_j=\mathrm{Max}\{h_j^{(R)},h_{j+1}^{(L)}\}$. 
We formally call $\{k_j\}_{j=0}^{{s}}$ critical points. 

When we obtain an explicit form of the inverse function {$h^{-1}(x)$ on each domain $(m_j,M_j)$}, 
then Eq.~(\ref{direct}) is one of the useful expressions for the density function. 
On the other hand, we propose that even if the inverse function {$h^{-1}(x)$} cannot be computed explicitly, 
one can apply the parametric expression described by Eq.~(\ref{parameter}) to find out some properties of the 
density function. 

By the way, in our case treated here, $\{x_{j,l}(k)\}_{j,l}$ satisfy the assumptions subjected to $h(k)$ in the above argument. 
All the {\it formally} critical points of $x_{j,l}(k)$ in $[-\pi,\pi)$ are arranged in the following table: 
\begin{center}
	\begin{tabular}{c|c|c|c}
	  & $j=0$ & $j=1$ & $j=2$ \\ \hline
        {\it formally} critical points & $\{0\}$ & $\{-\pi\}$ & $\{\pm \pi/2\}$
	\end{tabular}
\end{center}
and 
	\begin{equation*}
	\lim_{k\to \pm 0} x_{0,l}(k) = \mp (-1)^l\frac{1}{\sqrt{10}}, \;\;
        \lim_{k\to \pm\pi}x_{1,l}(k) = \mp (-1)^l\frac{1}{\sqrt{10}}
	\end{equation*}
	\begin{equation} \label{critical2} 
        \mathrm{and}\;\; x_{2,l}\left(\pm\frac{\pi}{2}\right) = \pm (-1)^l\frac{2}{7}.     
        \end{equation}
From Eq.~(\ref{critical2}), for $j\in \{0,1\}$ cases, it is useful to take the domain of the wave number $k$ 
by $k\in [0,2\pi)$ and $[-\pi, \pi)$, respectively. 
On the other hand, in case $j=2$, it is useful to decompose the domain $[-\pi,\pi)$ into $[-\pi/2,\pi/2)\cup [\pi/2,3\pi/2)$. 
Now we obtain the following proposition which implies Theorem~\ref{TM:th:convergence}. 
\begin{proposition}
	\begin{equation*}
	\lim_{t\to\infty} \mathbb{P}(X_t/t\leq x)=\int_{-\infty}^{x} \left\{\Delta\delta_0(y) + \sum_{m,l\in \{0,1\}}\rho_{m,l}(y)\right\} dy, 
	\end{equation*}
where the functions $\rho_{0,l}$ and $\rho_{1,l}$ have the following properties: 
\begin{enumerate}
	\item The functions $\rho_{0,l}$ and $\rho_{1,l}$ are continuous functions which have the supports $(-1/\sqrt{10},1/\sqrt{10})$ and $(-2/7,2/7)$, respectively. 
        \item The sets $C_{m}^{(l)}\equiv \{(x,\rho_{m,l}(x)): x\in\mathbb{R}\}$ $({m,l\in\{0,1\}})$ are described by the wave number $k$ as follows: 
	\begin{align}
        C_{0}^{(l)} &=
                \left\{ \left(x_{0,l}(k), \frac{w_{0,l}(k)+w_{1,l}(k-\pi)}{2\pi|d x_{0,0}(k)/d k|}\right): k\in [0,2\pi) \right\}, \label{C0}\\
        C_{1}^{(l)} &=
                \left\{ \left(x_{2,l}(k), \frac{w_{2,l}(k)+ w_{2,l}(\pi-k)}{2\pi|d x_{2,0}(k)/d k|}\right): k\in [-\pi/2,\pi/2) \right\}, \label{C1}
        \end{align}
where  
	\[ w_{j,l}(k)=|\langle v_{j,l}(k),\hat{\Psi}_0(k) \rangle|^2.  \]
\end{enumerate}
\end{proposition}
\begin{proof}
Noting that $\xi_0(k+\pi)=-\xi_1(k)$ and $\xi_2(\pi/2+k)=-\xi(\pi/2-k)$, we have 
\begin{align}
	x_{0,l}(k+\pi) &= x_{1,l}(k), \\
        x_{2,l}(\pi/2+k) &= x_{2,l}(\pi/2-k) \label{x}, \\ 
        \tau_{0,l}(k+\pi) &= \tau_{1,l}(k)=|d x_{0,0}(k)/d k|, \label{tau1}\\
        \tau_{2,l}(\pi/2+k) &= \tau_{2,l}(\pi/2-k)=|d x_{2,0}(k)/d k|. \label{tau2}
\end{align}
Here we put $\tau_{j,l}(k)=|d x_{j,l}(k)/d k|$. 
Equations~(\ref{x})-(\ref{tau2}) imply
\begin{multline}\label{w1}
	\left\{ \left(x_{1,l}(k), \frac{w_{1,l}(k)}{\tau_{1,l}(k)}\right): -\pi \leq k<\pi \right\} \\
        	= \left\{ \left(x_{0,l}(k), \frac{w_{1,l}(k-\pi)}{|d x_{0,0}(k)/d k|}\right): 0 \leq k< 2\pi \right\}, 
\end{multline}
\begin{multline}\label{w2}
        \left\{ \left(x_{2,l}(k), \frac{w_{2,l}(k)}{\tau_{2,l}(k)}\right): \pi/2 \leq k<3\pi/2 \right\} \\
                = \left\{ \left(x_{2,l}(k), \frac{w_{2,l}(\pi-k)}{|d x_{2,0}(k)/d k|}\right) : -\pi/2 \leq k< \pi/2 \right\}.
\end{multline}
Combining Eqs.~(\ref{w1}) and (\ref{w2}) with Eq.~(\ref{parameter}) leads to the desired conclusion of the second claim. 
By the way, we notice that Eq.~(\ref{critical2}) gives the support of $\rho_{j,l}(x)$. 
Moreover the continuity of $\rho_{j,l}(x)$ in the first claim is immediately obtained, 
since the first and second coordinates of $C_j^{(l)}$  $(j,l\in\{0,1\})$ are continuous with respect to parameter $k$ for each domain. 
\end{proof}
The terms $1/\tau_{j,l}(k)$ are independent of the initial state and give the support of the density function. 
On the other hand, the terms $w_{j,l}(k)$ depend on the initial condition.
Now let us put our focus on a limit density function whose behavior is expected by just the functions $1/\tau_{j,l}(k)$: 
there exists a random initial condition so that the function $w_{j,l}(k)$ averagely becomes constant. 
%
\begin{corollary}
We uniformly choose the initial state $\Psi_0$ from $\{|0\rangle \otimes |j\rangle: j\in\{0,1,\dots,9\}\}$; that is, 
	\[ Prob(``\Psi_0=|0\rangle\otimes |j\rangle")=1/10 \;\;\;\;(j\in\{0,1,\dots,9\}).\]  
Let the initial state be chosen randomly as the above. Then the density function is expressed by 
	\begin{equation*}
        \frac{2}{5}\delta_0(x)+\frac{3}{5}\left\{\nu_0(x)+\nu_1(x)\right\},
        \end{equation*}
where continuous functions $\nu_0(x)$ and $\nu_1(x)$ have the finite supports $(-1/\sqrt{10},1/\sqrt{10})$ and $(-2/7,2/7)$ 
and take infinity at the boundaries $|x|=1/\sqrt{10}$ and $|x|=2/7$, respectively. 
Moreover we have the following parametric expression. 
	\begin{align} 
        \{(x,\nu_0(x)):x\in \mathbb{R} \} = \left\{ \left(x_{0,0}(k), \frac{1}{{3}\pi |d x_{0,0}(k)/d k|}\right): 0 \leq k< 2\pi \right\}, \label{jJacobian1} \\
        \{(x,\nu_1(x)):x\in \mathbb{R} \} = \left\{ \left(x_{2,0}(k), \frac{1}{{3}\pi |d x_{2,0}(k)/d k|}\right): -\pi/2 \leq k<\pi/2 \right\}. \label{jJacobian2}
        \end{align}
\end{corollary}
\begin{proof}
We should notice the expectation with respect to the initial state provides 
	\begin{align*} 
        w_{j,l}(k) &= \mathbb{E}[ |\langle v_{j,l}(k),\hat{\Psi}_0(k)\rangle|^2 ]=\mathbb{E}[ |\langle v_{j,l}(k),\Psi_0\rangle|^2 ] \\
         &= \mathbb{E}[\mathrm{Tr}(|\Psi_0\rangle \langle \Psi_0|\cdot |v_{j,l}(k)\rangle \langle v_{j,l}(k)|)]  
          =\frac{1}{10} \mathrm{Tr}(|v_{j,l}(k)\rangle \langle v_{j,l}(k)|)\\
         &=\frac{1}{10}. 
        \end{align*}
Inserting them into Eqs.~(\ref{C0}) and (\ref{C1}) completes the proof. 
\end{proof}

Figure~\ref{fig:prob} shows probability distributions on a rescaled space by time $t$ when the walker 
starts from one of the points in $x=0$ at time $t=0$. 
In case (a) since $\ket{\Psi_0}=(\oplus_{m=0}^3\Gamma_m')^{\bot}$, then $\Delta=0$, 
on the other hand, in case (b), we have from Eq.~(\ref{Delta}), $\Delta=\sum_{m=0}^3|\langle \eta_{m;0}, \Psi_0\rangle|^2=1/2$. 
Moreover the function $f(x)$ is a linear combination of two continuous functions which have different finite supports each other. 
The shapes of the above continuous functions $\nu_0$ and $\nu_1$ are depicted in Fig.~\ref{Fig:densityfunctions}. 
\begin{center}
 \begin{figure}[h]
  \begin{minipage}{60mm}
   \begin{center}
    \includegraphics[scale=0.5]{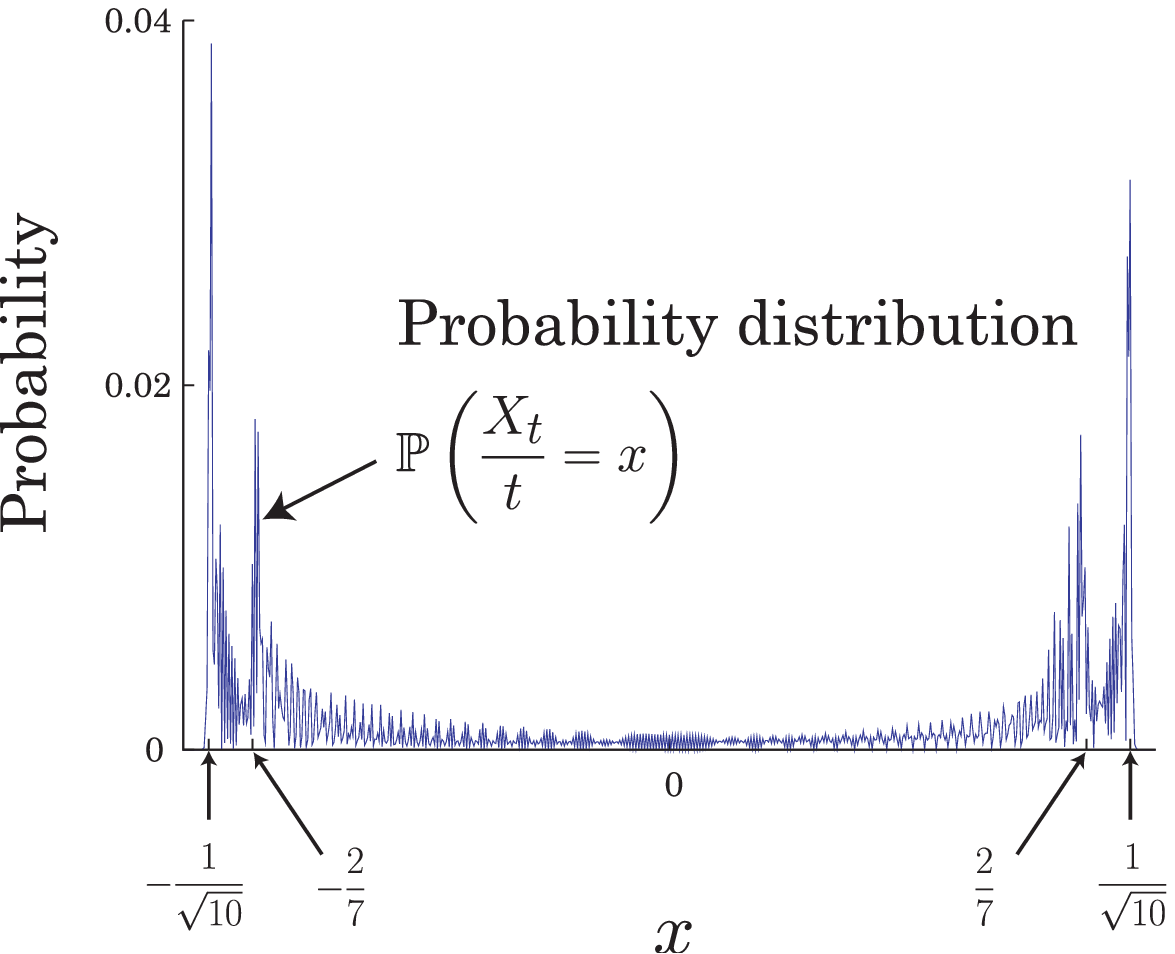}\\[3mm]
    {\scriptsize (a) $\ket{\Psi_0}=\ket{0}\otimes\frac{1}{\sqrt{3}}\left(\ket{7}+\ket{8}+\ket{9}\right)$} 
   \end{center}
  \end{minipage}\hspace{10mm}
  \begin{minipage}{60mm}
   \begin{center}
    \includegraphics[scale=0.5]{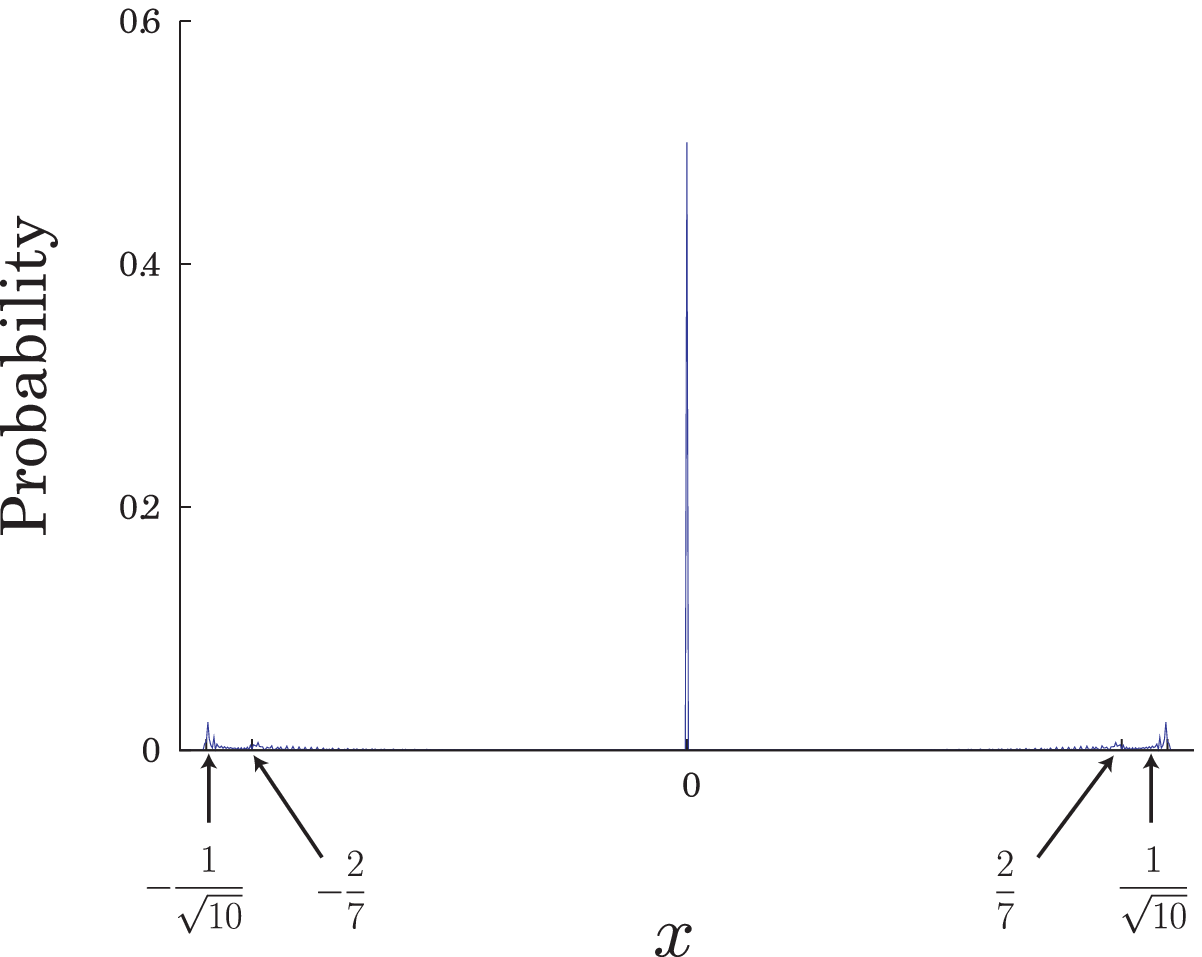}\\[3mm]
    {\scriptsize (b) $\ket{\Psi_0}=\ket{0}\otimes\frac{1}{\sqrt{2}}\left(\ket{3}+i\ket{4}\right)$}
   \end{center}
  \end{minipage}
  \vspace{3mm}
  \caption{Probability distributions $\mathbb{P}(X_t/t=x)$ at time $t=1000$}
  \label{fig:prob}
 \end{figure}
\end{center}

 \begin{figure}[h]
   \begin{center}
    \includegraphics[width=50mm]{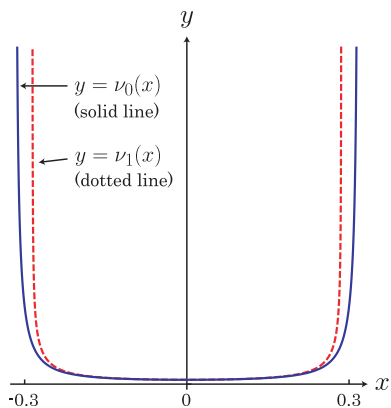}
   \end{center}
   \caption{The shapes of $\nu_0(x)$ and $\nu_1(x)$ by the exact computation of Eqs.~(\ref{jJacobian1}) and (\ref{jJacobian2}): 
   solid and dotted lines depict $\nu_0(x)$ and $\nu_1(x)$, respectively. }
   \label{Fig:densityfunctions}
 \end{figure}

\section{Discussion}
For comparing with our result, we provide well known result on 
the Szegedy walk based on asymmetric random walk on $\mathbb{Z}$ that a particle jumps to the left or the right with probability $q$ or $p$ $(pq\neq 0)$, 
respectively, at each time step. We denote the transition operator of the random walk $P_0$. 
Let {$Y_t^{RW;\mathbb{Z}}$} be the random walk at time $t$. 
The asymptotic of the variance of the random walk is $\lim_{t\to\infty}E[(Y_t^{RW;\mathbb{Z}}-E[Y_t^{RW;\mathbb{Z}}])^2/t]=\sigma^2$ with $\sigma=2\sqrt{pq}$. 
On the other hand, the spectrum of the time evolution operator of the Szegedy walk lies on a part of the unit circle $\{e^{i\theta}: |\cos \theta| \leq \sigma\}$ continuously. 
Let $Y_t^{QW;\mathbb{Z}}$ be a position of this QW at time $t$. The initial state is chosen uniformly from $\{\delta_{(0,1)},\delta_{(0,-1)}\}$. 
The density function for the weak convergence of $Y_t^{QW;\mathbb{Z}}/t$ is described by $f_K(x; \sigma)$ \cite{Konno2002,Konno2005}. 
For $p\neq q$ case, a parametric expression for the density function becomes 
	\begin{equation}\label{paraKonno}
        \left\{  \left( \pm \sqrt{\frac{\sigma^2-\lambda^2}{1-\lambda^2}}, \;
        \left|\frac{(1-\lambda^2)^{3/2}}{\pi (1-\sigma^2)\lambda}\right| \right): \lambda\in \mathrm{spec} (P_0) \right\}. 
        \end{equation}	
On the other hand, $p=q$ case, this walk corresponds to the potential free Schr{\"o}dinger equation, which implies  
	\begin{equation*}
        \left\{  \left( \lambda, \;\frac{1}{2}(\delta_{-1}(\lambda)+\delta_1(\lambda)) \right): \lambda\in \mathrm{spec} (P_0) \right\}. 
        \end{equation*}
Now let us consider our model. 
We can easily check that the three eigenvalues of the underlying twisted random walk are bounded by $2/3\leq \lambda_0(k)\leq 1$, $-1\leq \lambda_1(k) \leq -2/3$ and $-1/3\leq \lambda_2(k) \leq -1/3$. 
(see Eq.~(\ref{randomwalk}). ) For the transition matrix $P$ of the underlying random walk, it holds that 
	\[ \mathrm{spec}(P)=[-1,-2/3] \cup [-1/3,1/3] \cup [2/3,1] \cup\{ 0 \} . \]
Using the above, the parametric expression of the functions $\nu_0(x)$ and $\nu_1(x)$ in Eqs.~(\ref{jJacobian1}) and (\ref{jJacobian2}) is 
re-expressed by  
	\begin{equation}\label{para_ourmodel}
        \left\{  \left( \pm \sqrt{\frac{(\gamma^2-\lambda^2) \eta(\lambda)}{1-\lambda^2}},\; 
        \frac{7\sqrt{21}}{|G(\lambda)|\sqrt{1-\lambda^2}}\right): \lambda \in \mathrm{spec}(P) \right\}, 
        \end{equation}
where $\gamma=\sqrt{28/27}$ and $G(x)=F(x/\gamma)$, which is a rational function (see Eq.~(\ref{F}) for an explicit expression of $F(x)$).  
Here 
	\[ \eta(\lambda)=\frac{A^2-\cos^2\left\{ 3\arccos \left(\lambda/\gamma\right) \right\}}
		{9\sin^2\left\{ 3\arccos \left(\lambda/\gamma\right) \right\}}. \]
We notice that the value $1/\sqrt{10}$, which describes the support of the density function $f(x)$, 
corresponds to the standard deviation $\sigma$ of $Y_t^{RW;\mathbb{Z}}$ 
since this value $1/\sqrt{10}$ is equivalent to the coefficient of the first order of the standard deviation for the underlying random walk. 
From the above observations, it seems that a diffusion property of the underlying random walks is reflected to a spreading strength of the Szegedy walks. 

By the way, the spectrum of our QW is distributed on 
	\[\{e^{i \theta}: |\cos\theta|\leq 1/3\} \cup \{e^{i\theta}:|\cos\theta|\geq 2/3 \} 
        	\cup \{\pm i\} \cup \{\pm 1\}. \]
The first and second sets are the support of the absolutely continuous part, which are related to linear spreading, 
and the last two parts are its point spectrums, which are related to a homological structure and localization. 
The supports of the above two absolutely continuous parts are mutually disjoint. 
We observed that the split support of the spectrum provides a superposition of two kinds of continuous functions which have different finite supports each other. 
Another interesting point is that all the point spectrums are embedded in the support of the absolutely continuous parts. 
We have not found explicitly an effect of this embeddedness on the behavior of quantum walks. 
This is one of the interesting future's problem. 
%


\par
\
\par\noindent
\noindent
{\bf Acknowledgments.}
\par
TM is grateful to the Japan Society for the Promotion of Science for the support, and to the Math. Dept. UC Berkeley for hospitality.
ES thanks to the financial support of the Grant-in-Aid for Young Scientists (B) of Japan Society for the
Promotion of Science (Grant No.25800088). 

\par

\begin{small}
\bibliographystyle{jplain}

\begin{thebibliography}{99}
\bibitem{Ambainis}
A. Ambainis, 
\textit{Quantum walk algorithm for element distinctness}, 
Proc. 45th IEEE Symposium Foundations of Computer Science (2004) pp.22--31.

\bibitem{Kempe}
A. Ambainis, J. Kempe and A. Rivosh, 
\textit{Coins make quantum walks faster}, 
Proc. 33rd ACM Symposium on Theory of Computing (2005) pp.37--49.

\bibitem{Exner}
P. Exner and P. Seba, 
{\it Free quantum motion on a branching graph}, 
Rep. Math. Phys. {\bf 28} (1989) pp.7--26.

\bibitem{Konno2002}
N. Konno, 
\textit{Quantum random walks in one dimension}, 
Quantum Information Processing {\bf 1} (2002) pp.345--354.

\bibitem{Konno2005}
N. Konno, 
\textit{A new type of limit theorems for the one-dimensional quantum random walk}, 
Journal of the Mathematical Society of Japan {\bf 57} (2005) pp.1179--1195.

\bibitem{FH}
E. Feldman and M. Hillery, 
\textit{Quantum walks on graphs and quantum scattering theory}, 
Contemporary Mathematics {\bf 381} (2005) pp.71--96.

\bibitem{Gnutzmann}
S. Gnutzmann and U. Smilansky, 
{\it Quantum graphs: Applications to quantum chaos and universal spectral statistics}, 
Advances in Physics {\bf 55} (2006) pp.527--625.

\bibitem{GrimmettJansonScudo2004}
G. Grimmett, S. Janson and P.F. Scudo, 
\textit{Weak limits for quantum random walks}, 
Phys. Rev. E {\bf 69} (2004) 026119.

\bibitem{Grover}
L. K. Grover, 
\textit{A fast quantum mechanical algorithm for database search}, 
Proc. 28th ACM Symposium on the Theory of Computing {\bf 212} (1996) pp.212--219.

\bibitem{HKSS_2013}
Yu. Higuchi, N. Konno, I. Sato and E. Segawa, 
\textit{Quantum graph walks I: mapping to quantum walks}, 
Yokohama Mathematical Journal {\bf 59} (2013) pp.34--56.

\bibitem{HKSS_2014}
Yu. Higuchi, N. Konno, I. Sato and E. Segawa, 
\textit{Spectral and asymptotic properties of Grover walks on crystal lattices}, 
arXiv:1401.0154.

\bibitem{SchanzSmilansky}
H. Schanz and U. Smilansky, 
{\it Periodic-orbit theory of Anderson localization on graphs, }
Physical Review Letters {\bf 14} (2000) pp.1427--1430.

\bibitem{Shenvi}
N. Shenvi, J. Kempe and B. Whaley, 
\textit{A quantum random walk search algorithm}, 
Phys. Rev. A {\bf 67} (2003) 052307.

\bibitem{Sunada}
T. Sunada, 
{\it Topological Crystallography}, 
Surveys and Tutorials in the Applied Mathematical Sciences {\bf 6} (2013), Springer. 

\bibitem{Sze}
M. Szegedy, 
{\it Quantum speed-up of Markov chain based algorithms}, 
Proc. 45th IEEE Symposium on Foundations of Computer Science (2004) pp.32--41. 

\bibitem{Tanner}
G. Tanner, 
{\it From quantum graphs to quantum random walks}, 
Non-Linear Dynamics and Fundamental Interactions NATO Science Series II: Mathematics, Physics and Chemistry {\bf 213} (2006) pp.69--87.

\bibitem{VA1} 
S.E. Venegas-Andraca and J.L Ball, 
{\it Processing Images in Entangled Quantum Systems, }
Quantum Information Processing {\bf 9} (2010) pp.1--11. 

\bibitem{VA2} 
S.E. Venegas-Andraca and S. Bose, 
{\it Storing, Processing, and Retrieving an Image using Quantum Mechanics, }
Proc. SPIE Conference on Quantum Information and Computation (2003) pp.137-147.

\bibitem{Watrous}
J. Watrous, 
{\it Quantum simulations of classical random walks and undirected graph connectivity, }
Journal of Computer and System Sciences {\bf 62} (2001) pp.376--391.

\end{thebibliography}

\end{small}


\end{document}